%% file: main.tex
\documentclass[11pt, a4paper]{article}

\input{defs}

\begin{document}

\title{Cryptanalysis of a System Based on Twisted Reed--Solomon Codes}

\author{Julien Lavauzelle\thanks{Université de Rennes, CNRS, IRMAR -- UMR 6625, France. Email: {\tt julien.lavauzelle@univ-rennes1.fr}
  }
  \and
  Julian Renner\thanks{Institute for Communications Engineering, Technical University of Munich (TUM), Germany. Email: {\tt julian.renner@tum.de}
  }
}

\date{\today}

\maketitle              


\begin{abstract}
  Twisted Reed--Solomon (TRS) codes are a family of codes that contains a large number of maximum distance separable codes that are non-equivalent to Reed--Solomon codes. 
  TRS codes were recently proposed as an alternative to Goppa codes for the McEliece code-based cryptosystem, resulting in a potential reduction of key sizes. \correct{The use of TRS codes in the McEliece cryptosystem has been motivated by the fact that a large subfamily of TRS codes is resilient to a direct use of known algebraic key-recovery methods.}
  
  In this paper, an efficient key-recovery attack on the TRS variant that was used in the McEliece cryptosystem is presented. The algorithm \correct{exploits a new approach based on recovering the structure of a well-chosen subfield subcode of the public code.}   It is proved that the attack always succeeds and  breaks the system for all practical parameters in $O(n^4)$ field operations. A software implementation of the algorithm 
  retrieves a valid private key from the public key within a few minutes, for parameters claiming a security level of $128$ bits. \correct{The success of the attack also indicates that, contrary to common beliefs, subfield subcodes of the public code need to be precisely analyzed when proposing a McEliece-type code-based cryptosystem.}
Finally, the paper discusses an attempt to repair the scheme and a modification of the attack aiming at Gabidulin--Paramonov--Tretjakov cryptosystems based on twisted Gabidulin codes.

\end{abstract}

\section{Introduction}\label{sec:intro}

In the last years, cryptosystems relying on the hardness of decoding in a generic code have gained a lot of attention due to their potential resistance against quantum computer attacks. The first code-based cryptosystem was proposed by McEliece already in 1978~\cite{mceliece1978public}.  \correct{Its hardness is based on the assumption that a random generator matrix of a random binary Goppa code is hard to distinguish from the generator matrix of a random code.} To this day, the principle behind the McEliece system still plays a significant role in the design of code-based cryptography. In particular, four out of the six code-based proposals in round 2 of the National Institute of Standards and Technology (NIST) post-quantum cryptography standardization process are based on McEliece's principle.

\correct{Compared to other post-quantum-secure public-key encryption schemes, \emph{e.g.}\ some lattice-based cryptosystems,} the main drawback of the McEliece cryptosystem lies in the size of its public key. To overcome this drawback, other families of codes have been proposed to replace Goppa codes, but most of them can be subjected to algebraic attacks. For instance, generalized Reed--Solomon (GRS) codes were proposed in 1986 by Niederreiter~\cite{Niederreieter1986}, but Sidelnikov and Shestakov mounted a very efficient attack to recover an alternative secret key~\cite{sidelnikov1992attack}. Wieschebrink proved that also random subcodes of GRS codes --- proposed in~\cite{Berger2005} --- cannot be used due to their vulnerability to the \emph{code squaring} attack~\cite{Wieschebrink2010}. Further instances and cryptanalyses of algebraic code-based schemes can be found in~\cite{Sidelnikov94, MinderS07, BergerCGO09,FaugereOPPT16, JanwaM96, CouvreurCP17}. \correct{One should emphasize that many recent attacks are largely based on previously known methods. For example, some instances of the RLCE scheme~\cite{Wang16} were broken by Couvreur, Lequesne and Tillich by a sophisticated analysis of the squares of \emph{puncturings and shortenings} of the public code~\cite{CouvreurLT19}.}

One recent alternative class of codes for the McEliece cryptosystem emerged from twisted Reed--Solomon (TRS) codes~\cite{beelen2017twisted}. In particular, Beelen \emph{et al.} analyzed the structural properties of a very specific subfamily of TRS codes~\cite{beelen2018structural}. They proved that this subfamily is disjoint from the class of GRS codes; thus the attack by Sidelnikov and Shestakov~\cite{sidelnikov1992attack} cannot be applied to their system. Further, they showed that shortenings of these codes up to two positions have maximal Schur square dimension~\cite{Puchinger_Diss}, meaning that the proposed system is impervious to a direct application of the attack presented by Couvreur \emph{et al.} in~\cite{Couvreur2014}. Additionally, the authors gave evidence that their system is not vulnerable to {straight-forward} applications of methods introduced by Wieschebrink in~\cite{wieschebrink2006attack,Wieschebrink2010}.

The intention of the authors of~\cite{beelen2018structural} was to exploit the optimal error-correction capability of TRS codes to reduce the length of the public code, and accordingly the size of the public key. In~\cite{beelen2018structural}, an explicit subfamily of TRS codes was proposed, providing a reduction of the public key up to a factor of $7.4$ compared to binary Goppa codes, for a claimed security level of $128$ bits.

In this paper, we present an \correct{efficient key-recovery attack on} this cryptosystem based on TRS codes. \correct{As analyzed by the authors of~\cite{beelen2018structural}, the direct application of previously known structural attacks does not work. Instead, we recover the structure of a well-chosen \emph{subfield subcode} $\mathcal{S}$ of the public TRS code $\mathcal{T}$. We give a characterization of the structure of this subfield subcode, as a subspace of low codimension contained in a classical Reed--Solomon code $\mathcal{R}$. We then prove that the Wieschebrink squaring method \emph{always} succeeds when applied to the subfield subcode $\mathcal{S}$, and this enables us to retrieve an algebraic description of $\mathcal{R}$. By analyzing equivalent representations of TRS codes, we finally deduce an algebraic description of the public code $\mathcal{T}$. The application of the squaring method to the subfield subcode is a non-trivial modification of Wieschebrink's attack.}

\correct{To the best of our knowledge, our attack is the first of its kind to exploit structural weaknesses of \emph{subfield subcodes} of the public code. On the contrary, the restriction to a subfield is usually considered as an operation that breaks the structure of an algebraic code and therefore makes it suitable for cryptography as attested by the attack-resilience of Goppa codes despite being subfield subcodes of Reed--Solomon codes. Our approach of attacking the subfield subcode instead of the original code might also be applicable to other classes of codes used in code-based cryptography.}

We show that for all practical parameters proposed by the designers, our algorithm recovers a valid private key from the public key in $O(n^4)$ operations over the underlying field, where $n$ denotes the code length. The attack is implemented in the computer-algebra system SageMath~\cite{sagemath} and is made public. Although the implementation is not optimized, it determines a valid private key in approximately two minutes for the parameters proposed in~\cite{beelen2018structural}.

The paper is structured as follows. In Section~\ref{sec:preliminaries}, we introduce the notation, and state the definition as well as important structural properties of TRS codes. In Section~\ref{sec:trssystem}, we present the key generation, encryption and decryption algorithm and the parameters proposed in~\cite{beelen2018structural}. In Section~\ref{sec:attack}, we derive a structural attack on the scheme and we precisely analyze its complexity. Additionally, in Section~\ref{sec:discussion}, we discuss a potential fix of the cryptosystem, as well as an extension of the attack to the rank metric setting~\cite{puchinger2018twisted}. Conclusions are given in Section~\ref{sec:conclusion}.

\section{Preliminaries}
\label{sec:preliminaries}

\subsection{Notation}
Let $\Fq$ denote the finite field of order $q$, where $q$ is a prime power.
Vectors in $\Fq^n$ are row vectors, and we use $\Fq^{m \times n}$ to represent the set of $m\times n$ matrices over $\Fq$.  For $i \in  \{1, \dots, m\}$ and $j \in \{1,\dots, n\}$, the $(i,j)$-th entry of $\bfA \in \Fq^{m \times n}$ is denoted by $A_{i,j}$. The set of invertible matrices of size $m$ over $\FF_q$ is denoted by $\GL_m(\FF_q)$.

Let us fix a finite field extension $\F/\FF_q$. The $\FF$-vector space generated by a subset $S \subset \FF_q^n$ is denoted by $\Span_\FF(S)$. By convention, we also represent the $\FF$-vector space spanned by the rows of $\bfA \in \FF_q^{m \times n}$ by $\Span_\FF(\bfA)$.

A linear code $\calC \subseteq \FF_q^n$ with parameters $[n, k, d]$ is an $\FF_q$-vector space of $\FF_q^n$ of dimension $k$, where $d$ is the minimum Hamming weight $\wt(\bfc) \mydef |\{ i \in \{1,\dots, n\}, c_i \ne 0\} |$ of a non-zero codeword $\bfc \in \calC$. A generator matrix of $\calC$ is a matrix $\bfG \in \FF_q^{k \times n}$ such that $\calC = \Span_{\FF_q}(\bfG).$

Given $\a \in \Fq^n$ and $\b \in \Fq^n$, their componentwise product is defined as $\a \star \b := ( a_1b_1, \hdots, a_nb_n  ) \in \Fq^n$.  Further, we define the Schur-square (or Hadamard-square) of a linear code $\mathcal{C} \subseteq \Fq^n$ as

\begin{equation*}
\mathcal{C}^{(\star 2)} := \Span_{\FF_q}(\{ \a \star \b : \a,\b \in \mathcal{C} \} ).
\end{equation*}

Let $\FF_q[X]$ denote the set of univariate polynomials over $\FF_q$. For a fixed evaluation vector $\alphaVec = ( \alpha_1, \hdots, \alpha_n ) \in \Fq^n$, we define the evaluation map
\[
\begin{array}{rclc}
\evOp_{\alphaVec} : &\Fq[X] &\rightarrow &\Fq^n \\
&f &\mapsto &( f(\alpha_1),f(\alpha_2), \dots, f(\alpha_n) ).
\end{array}
\]

Finally, if $\mathcal{I}, \mathcal{J} \subset \NN$ are two finite subsets of integers, then we define their sumset
\begin{equation*}
  \mathcal{I} + \mathcal{J} := \{a+b \, : \, a\in\mathcal{I},~b\in\mathcal{J} \} \subseteq \NN.
\end{equation*}

\subsection{Twisted Reed--Solomon codes}\label{subsec:trs}

Before introducing TRS codes, let us first recall the definition of (classical) Reed--Solomon codes.

\begin{definition}[Reed--Solomon code]\label{def:rs_definition}
  Let the entries of $\alphaVec = (\alpha_1,\dots,\alpha_n) \in \Fq^n$ be pairwise distinct, and fix $1 \le k \le n$. The \emph{Reed--Solomon (RS) code} of length $n$ and dimension $k$ is defined by
  \begin{equation*}
    \RS_{k,n}[\alphaVec] := \left\{ \ev{f}{\alphaVec} : f \in \FF_q[X], \deg f \le k-1 \right\} \subseteq \FF_q^n.
  \end{equation*}
The entries of $\alphaVec$ are called \emph{locators} of the Reed--Solomon code $\RS_{k,n}[\alphaVec]$.
\end{definition}

RS codes are maximum distance separable (MDS) codes, i.e., they reach the Singleton bound $d \le n- k +1$. They also admit the use of efficient decoding algorithms for an error of weight up to the unique decoding radius $\lfloor \frac{n-k}{2} \rfloor$.

TRS codes were recently constructed as a generalization of RS codes~\cite{beelen2017twisted}. Let us first define a specific subspace of polynomials. Let $\ell \ge 1$, and $n \ge k \ge 1$. Given a vector $\hVec \in \{0,\dots,k-1\}^\ell$ of pairwise distinct increasing \emph{hooks}, a vector $\tVec \in \{1, \dots, n-k\}^\ell$ of pairwise distinct \emph{twists}, and a vector of field coefficients $\etaVec \in (\Fq \setminus \{0\})^\numTwists$, the set of $[\tVec, \hVec, \etaVec]$-twisted polynomials is
\[
\calP_{k,n}[\tVec, \hVec, \etaVec] \mydef \left\{ \sum_{i=0}^{k-1} f_i X^i + \sum_{j=1}^{\numTwists} \eta_j f_{h_j} X^{k-1+t_j} : f_i \in \Fq \right\} \subseteq \F_q[X].
\]

\begin{definition}[Twisted Reed--Solomon code,~\cite{beelen2017twisted}]\label{def:trs}
  Let the entries of $\alphaVec = ( \alpha_1,\dots,\alpha_n ) \in \FF_q^n$ be pairwise distinct, and fix $1 \le k \le n$. Let $\tVec, \hVec, \etaVec$ be defined as above. The \emph{$[\tVec, \hVec, \etaVec]$-twisted Reed--Solomon (TRS) code} of length $n$, dimension $k$ and locators $\alphaVec$ is defined by
  \begin{equation*}
    \TRS_{k,n}[\alphaVec, \tVec, \hVec, \etaVec] := \left\{ \ev{f}{\alphaVec}  \, : \, f \in \calP_{k,n}[\tVec, \hVec, \etaVec] \right\}.
  \end{equation*}
\end{definition}

According to Definition~\ref{def:trs}, a generator matrix of $\TRS_{k,n}[\alphaVec, \tVec, \hVec, \etaVec]$ is given by
\begin{equation*}
  \Gtrs :=
  \begin{mymatrix}
    \ve{1} \\
    \alphaVec^{1} \\
    \vdots \\
    \alphaVec^{h_1-1} \\
    \alphaVec^{h_1} + \eta_1 \alphaVec^{k-1+t_1} \\
    \alphaVec^{h_1+1} \\
    \vdots \\
    \alphaVec^{h_\ell-1} \\
    \alphaVec^{h_\ell} + \eta_\ell \alphaVec^{k-1+t_\ell} \\
    \alphaVec^{h_\ell+1} \\
    \vdots \\
    \alphaVec^{k-1}
  \end{mymatrix},
\end{equation*}
where $\alphaVec^i := (\alpha_1^i, \dots, \alpha_n^i)$ for $i =1,\dots,k-1$.

In~\cite{beelen2018structural}, the authors show that the construction of TRS codes according to Definition~\ref{def:trs} does not necessarily lead to MDS codes. However, they provide a method to obtain a subfamily of MDS TRS codes, cf. Theorem~\ref{theorem:trsmds}.

\begin{theorem}[Explicit MDS TRS codes~\cite{beelen2018structural}]\label{theorem:trsmds}
  Let $q_0$ be a prime power, and $1 = s_0 < \hdots < s_{\ell} \in \mathbb{Z}_{>0}$ be non-negative integers such that $\F_{q_0^{s_0}} \subset \mathbb{F}_{q_0^{s_1}} \subset \hdots \subset \mathbb{F}_{q_0^{s_\ell}} = \Fq$ is a chain of subfields.
  Fix $k < n \leq q_0$ and the entries of $\alphaVec = (\alpha_1, \hdots, \alpha_n) \in \FqZ^n$ as pairwise distinct locators.
  Finally, let $\tVec$, $\hVec$ and $\etaVec$ be chosen as in Definition~\ref{def:trs}, such that $\eta_i \in \mathbb{F}_{q_0^{s_i \vphantom{-1}}}\setminus \mathbb{F}_{q_0^{s_i-1}}$ for $i=1,\hdots,\ell$.
  Then $\TRS_{k,n}[\alphaVec, \tVec, \hVec, \etaVec]$ is MDS.
\end{theorem}

A decoding algorithm for TRS codes is also proposed in~\cite{beelen2018structural}. Given a corrupted codeword $\ve{r} = \bfc + \bfe \in \FF_q^n$, where $\bfc \in \TRS_{k,n}[\alphaVec, \tVec, \hVec, \etaVec]$, the strategy is to guess $\ell$ elements $g_1,\hdots,g_{\ell} \in \Fq$ and then to decode $\ve{r} - \ev{\sum_{i=1}^{\ell} g_i \eta_i X^{t_i+k-1} }{\alphaVec}$ in the Reed--Solomon code $\RS_{k,n}[\alphaVec]$. This approach succeeds if $g_i = f_{h_i}$ and thus admits a worst case complexity in $O(q^\ell n \log^2 n \log \log n)$. Notice that for the explicit family presented in Theorem~\ref{theorem:trsmds}, we have $q = \Omega(q_0^{2^\ell})$, hence this decoding algorithm is only practical for a tiny number of twists.

The following lemma shows that TRS codes are invariant under specific transformations of their parameters. This property is a key element for the cryptanalysis of the system, and could be of independent interest.

\begin{lemma}\label{lemma:equiCode}
 Let $\alphaVec$, $\tVec$, $\hVec$ and $\etaVec$ be defined as in Definition~\ref{def:trs}. Then for any $a \in \Fq\setminus \{0\}$,
 \begin{equation*}
\TRS_{k,n}[\alphaVec, \tVec, \hVec, \etaVec] = \TRS_{k,n}[\hat{\alphaVec}, \tVec, \hVec, \hat{\etaVec}],
 \end{equation*}
 where $\alphaVecH = a \alphaVec$ and $\etaVecH = ( \hat{\eta}_1, \hdots, \hat{\eta}_{\ell} )$ with $ \hat{\eta}_i = \eta_i a^{-(k-1+t_i-h_i)}$, $1\leq i\leq\ell$.
\end{lemma}

\begin{proof}
  Let $\ev{f}{\alphaVecH} \in \TRS_{k,n}[\hat{\alphaVec}, \tVec, \hVec, \hat{\etaVec}]$, where $f(X) = \sum_{i=0}^{k-1} f_i X^i + \sum_{j=1}^{\ell} \hat{\eta}_j f_{h_j} X^{k-1+t_j}$. We have
  \[
  f(a X) = \sum_{i=0}^{k-1} (f_i a^i) X^i + \sum_{j=1}^{\ell} (\hat{\eta}_j a^{k-1+t_j-h_j}) (f_{h_j}a^{h_j}) X^{k-1+t_j} = g(X)\,,
  \]
  where $g(X) \in \calP_{k,n}[\tVec, \hVec, \etaVec]$. Hence by definition $\ev{f}{\alphaVecH} \in \TRS_{k,n}[\alphaVec, \tVec, \hVec, \etaVec]$, and it follows that $\TRS_{k,n}[\hat{\alphaVec}, \tVec, \hVec, \hat{\etaVec}] \subseteq  \TRS_{k,n}[\alphaVec, \tVec, \hVec, \etaVec]$. The proof on the converse inclusion is similar since $a$ is non-zero. \hfill~\qed
\end{proof}

\section{The variant of the McEliece cryptosystem using TRS codes}
\label{sec:trssystem}

 \subsection{Definition of the cryptosystem}
 \label{subsec:definition-cryptosystem}

 \paragraph{Setup.} Fix a prime power $q_0$, and integers $k < n \leq q_0-1$ with $2\sqrt{n}+6<k \leq \frac{n}{2}-2$. Also fix $\ell \in \mathbb{Z}_{>0}$ satisfying
 \[
 \frac{n+1}{k-\sqrt{n}}< \ell + 2 < \min\Big\{k+3;\frac{2n}{k}; \sqrt{n}-2 \Big\}\,.
 \]
 Further, set $q_i := q_{i-1}^2 = q_0^{2^i}$ for $ i= 1,\hdots, \ell$, such that $ \FqZ \subset \FqO \subset \hdots \subset \Fqell = \Fq$
 is a chain of subfields. Finally, set $t_i = (i+1)(r-2)-k+2$ and $h_i = r-1+i$ for $i =1,\hdots, \ell$,  where $r:= \lceil \frac{n+1}{\ell+2} \rceil +2$.

 Integers $q_0$, $n$, $k$, $\ell$, and tuples $\ve{t}$, $\ve{h}$ satisfying the above conditions are referred to as \emph{valid parameters} of the cryptosystem~\cite{beelen2018structural}. They are public parameters.

\paragraph{Key generation.} Given valid parameters  $q_0$, $n$, $k$, $\ell$,  $\ve{t}$ and $\ve{h}$, a pair of public/private keys is generated as follows.
\begin{enumerate}
\item Choose $\alphaVec \in \FqZ^{n}$ at random such that the entries of $\alphaVec$ are pairwise distinct.
\item Choose $\etaVec \in \Fq^{\ell}$ at random such that $\eta_i \in \Fqi \setminus \Fqii$ for $i = 1, \dots, \ell$.
\item Choose $\S \in \GL_k(\FF_q)$ at random.
\item Output the public key $\Gpub = \S \Gtrs \in \Fq^{k \times n}$, where $\Gtrs$ is the generator matrix of $\TRS_{k,n}[\alphaVec, \tVec, \hVec, \etaVec]$ described in Section~\ref{subsec:trs}.
\end{enumerate}
The private key consists of $(\S, \alphaVec, \etaVec)$ and the public key is $\Gpub$.

\paragraph{Encryption.} Given a plaintext $\plain \in \Fq^{k}$ and the public key $\Gpub$, the ciphertext is generated as follows.
\begin{enumerate}
\item Choose $\e \in \Fq^{n}$ at random with Hamming weight $\wt(\e) = \lfloor \frac{n-k}{2} \rfloor$.
\item Output the ciphertext
  \begin{equation*}
\cipher \mydef \plain \Gpub + \e \in \Fq^n.
    \end{equation*}
\end{enumerate}

\paragraph{Decryption.} Given a ciphertext $\cipher \in \Fq^n$ and the private key $(\S,\alphaVec,\etaVec)$, the decryption algorithm can be described as follows.
\begin{enumerate}
\item Decode $\cipher$ to $\plainT = \plain \S \in \Fq^{k}$ using the decoding algorithm of $\TRS_{k,n}[\alphaVec, \tVec, \hVec, \etaVec]$ given in~\cite{beelen2018structural}.
\item Output the plaintext $\plain = \plainT \S^{-1}$.
\end{enumerate}

\paragraph{Proposed Parameters.} The designers of the system proposed the parameters listed in Table~\ref{tab:proposed-param}~\cite{beelen2018structural}. Recall that the public code is defined over the field $\FF_q = \FF_{q_0^{2^\ell}}$.
\begin{table}[h!]
  \centering
  \begin{tabular}{cccccc}
    \hline
    $q_0$ & $n$    & $k$ & $\ell$ & $\ve{t}$ & $\ve{h}$ \\
    \hline
    $256$ & $255$ & $117$ & $1$   & $(57)$   & $(88)$ \\
    \hline
  \end{tabular}
  \caption{\label{tab:proposed-param}Parameters proposed in~\cite{beelen2018structural} for a claimed security $\ge 100$ bits.}
\end{table}

There are two main reasons for choosing a small number of twists. On the one hand, the complexity of the decoding algorithm proposed in~\cite{beelen2018structural} is in $O( q_0^{\ell2^{\ell}} n \log^2 n \log \log n)$ and thus increases doubly exponentially with the number of twists. On the other hand, the number of elements in the largest field $\FF_q$ also scales exponentially with the number of twists, which impacts the key sizes.

\correct{
\subsection{Resistance to some known key-recovery algebraic attacks}
}

\correct{As mentioned in Section~\ref{sec:intro}, Beelen \emph{et al.} showed that some existing attacks cannot be \emph{directly} mounted on their system~\cite{beelen2018structural}. Let us recall these attacks and explain why they are ineffective.}

\correct{
\paragraph{Sidelnikov--Shestakov attack.} In~\cite{sidelnikov1992attack}, Sidelnikov and Shestakov presented an attack on a variant of the McEliece cryptosystem using GRS codes. The attack uses two key facts: first, for MDS codes it is easy to find minimal-weight codewords with a given support, by running a simple Gaussian elimination; second, the ratio between two minimial-weight codewords of a GRS code, whose supports differ in only two coordinates, gives a rational function of degree one. Using these properties, the recovery of an alternate public key (\emph{i.e.} an algebraic description of the public code as a GRS code) reduces to solving linear systems of equations involving the coefficients of the rational functions and the parameters of the GRS code. Formally, the result of Sidelnikov and Shestakov~\cite{sidelnikov1992attack} can be summarized as follows.
\begin{theorem}[Sidelnikov--Shestakov~\cite{sidelnikov1992attack}]\label{thm:sidelnikov-shestakov}
  Let $\RS_{k,n}[\alphaVec]$ be a Reed--Solomon code with locators $\alphaVec = (\alpha_1,\hdots,\alpha_n) \in \FqZ^n$. Given any generator matrix of $\RS_{k,n}[\alphaVec]$, there exists an algorithm which determines in time $O(n^4)$ a vector $\alphaVecP \in \FqZ^{n}$ such that
\begin{equation*}
\RS_{k,n}[\alphaVec] = \RS_{k,n}[\alphaVecP].
\end{equation*}
In particular, it holds that $\alphaVecP = a \alphaVec + b \ve{1} := ( a \alpha_1 + b, \dots, a\alpha_n + b )$ with $a\in\FqZ\setminus\{0\}$ and $b\in\FqZ$.
\end{theorem}
}

\correct{However for TRS codes, the ratio of two minimal-weight codewords with close support is a high degree rational function involving many coefficients. This property prevents a direct use of Sidelnikov--Shestakov's attack.}
\correct{
  \paragraph{Wieschebrink attack.}
  In order to attack a variant of McEliece cryptosystem using random subcodes of GRS codes, Wieschebrink considered the following structural properties. Let $\calC$ be a random subcode of dimension $k-m$ of a GRS code of dimension $k$, with $m$ small compared to $k$. With high probability, the Schur square $\calC^{(\star 2)}$ is a GRS code of dimension $\min \{ n, 2k-1 \}$. If $k < n/2$, a Sidelnikov--Shestakov attack can be applied to recover the secret parameters. Otherwise, one can shorten the public code to fulfill the latter condition, since a shortened RS code is again a RS code.
}

\correct{
  As proved by the designers of the cryptosystem, Wieschebrink's idea cannot be directly applied to TRS codes, due to a smart choice of parameters: the Schur square of the public code has dimension $n$, and shortening techniques seem unappropriate since the family of TRS codes is not stable under this operation. We will see in the following section that restricting TRS codes to subfields however leaks the algebraic structure of the public code.}

\section{An efficient key-recovery attack using subfield subcodes}\label{sec:attack}

This section presents an efficient key-recovery algorithm for the cryptosystem with the parameters proposed in~\cite{beelen2018structural}. The algorithm first determines a linear transformation of the secret locators $\alphaVec$ by exploiting structural properties of the \emph{subfield subcode} of the public code. Then, the algorithm finds the coefficients of the twist monomials by Lagrange interpolation. The algorithm finally outputs $(\hat{\S},\alphaVecH,\etaVecH)$ such that $\hat{\S}\GtrsH = \Gpub$. As shown in Lemma~\ref{lemma:equiCode}, $(\hat{\S},\alphaVecH,\etaVecH)$ is a valid private key that can be used in the decryption algorithm (see Section~\ref{subsec:definition-cryptosystem}).

\subsection{Key-recovery algorithm}
\label{sec:key-recovery}

\subsubsection{First step: recovery of an affine transformation of the secret locators}
Let us consider the $\F_{q_0}$-subfield subcode of the code $\mathcal{C}_{\rm pub}$ spanned by the public generator matrix $\Gpub$. We first state a technical lemma.
\begin{lemma}
  \label{lem:subfield-interpolation}
  Let the entries of $\alphaVec = (\alpha_1, \dots, \alpha_n) \in \F_{q_0}^n$ be pairwise distinct. Further, let $P \in \F_q[X]$ where $\F_q$ is an extension of $\F_{q_0}$, such that $\deg(P) < n$. Then, $\evOp_{\alphaVec}(P) \in \F_{q_0}^n$ if and only if $P \in \F_{q_0}[X]$.
\end{lemma}
\begin{proof}
  Let $\c = \evOp_{\alphaVec}(P)$ and assume that $\c \in \F_{q_0}^n$. Since $\alphaVec \in \F_{q_0}^n$ and $n \le q_0$, there exists a polynomial $Q \in  \F_{q_0}[X]$ of degree $\le n$ such that $\c = \evOp_{\alphaVec}(Q)$. Moreover, $\evOp_{\alphaVec}$ is injective over the $\F_q$-subspace of polynomials of degree $<q_0$, hence $P=Q$. The converse is straightforward. \hfill~\qed
\end{proof}

Let us now define $\Ifree:=\{0,1,\hdots,k-1\}\setminus\{h_1,\hdots,h_{\ell}\}$ as the set of exponents of monomials which do not support twists\footnote{\correct{Since the parameters $k$ and $h_1,\hdots,h_{\ell}$ are public, an attacker knows the set $\mathcal{I}$.}}. \correct{For valid parameters, $\Ifree = \{0,1, \dots, r-1 \} \cup \{r+\ell,\hdots,k-1\}$ since $h_i = r-1+i$ for each $1 \le i \le \ell$.} We can now prove the following characterization of subfield subcodes of TRS codes with valid parameters.

\begin{proposition}\label{prop:subcode}
  Let $\TRS_{k,n}[\alphaVec, \tVec, \hVec, \etaVec]$ be chosen with valid parameters, as described in Section~\ref{sec:trssystem}. Define $\mathcal{I} = \{0,1,\hdots,k-1\}\setminus\{h_1,\hdots,h_{\ell}\}$ as above. Then,
\begin{equation*}
  \TRS_{k,n}[\alphaVec, \tVec, \hVec, \etaVec] \cap \FqZ^n = \Span_{\FF_{q_0}}\big( \{ \ev{X^i}{\ve{\alpha}}, i \in \Ifree \}\big)\,.
\end{equation*}
\end{proposition}

\begin{proof}
Let us denote $\mathcal{S} = \Span_{\FF_{q_0}}\big( \{ \ev{X^i}{\ve{\alpha}}, i \in \Ifree \}\big)$ and $\Cpub = \TRS_{k,n}[\alphaVec, \tVec, \hVec, \etaVec]$.  
  First, it is clear that $\mathcal{S} \subseteq \Cpub \cap \FqZ^n$. Indeed, for $i \in \Ifree$ we have $\ev{X^i}{\alphaVec} \in \Cpub$, and since $\alphaVec$ is a vector over $\FqZ$, it yields that $\ev{X^i}{\alphaVec} \in \FqZ^n$.
  
  Conversely, let $\c = \ev{f}{\alphaVec} \in \Cpub \cap \FqZ^n$, where $f \in \calP_{k,n}[\tVec, \hVec, \etaVec]$. Lemma~\ref{lem:subfield-interpolation} ensures that $f \in \F_{q_0}[X]$, since $\deg(f) < n$ for valid parameters. It remains to notice that $\F_{q_0}[X] \cap \calP_{k,n}[\tVec, \hVec, \etaVec] = \Span_{\FF_{q_0}}( \{ X^i, i \in \Ifree \})$. \hfill~\qed
\end{proof}

\correct{
We observe by Proposition~\ref{prop:subcode} that the subfield subcode $\Csub := \TRS_{k,n}[\alphaVec, \tVec, \hVec, \etaVec] \cap \FqZ^n$ is a proper non-MDS subcode of the RS code $\RS_{k,n}[\alphaVec]$. Thus, one cannot directly use a Sidelnikov--Shestakov attack~\cite{sidelnikov1992attack} on $\Csub$. In 2006, Wieschebrink mounted an attack on cryptosystems based on \emph{random} subcodes of RS codes~\cite{Wieschebrink2010}. The author's idea is that, with very high probability over the chosen subcode $\mathcal{C}'$, the square code $\mathcal{C}'^{(\star 2)}$ is a RS code. A Sidelnikov--Shestakov attack can then be used on $\mathcal{C}'^{(\star 2)}$ to recover the private parameters.
}

\correct{
In the following, we prove that for most valid parameters defined in~\cite{beelen2018structural}, and for \emph{all} practical ones, the square code $\Csub^{(\star 2)}$ is a RS code subject to a Sidelnikov--Shestakov attack.}
\begin{proposition}\label{prop:squareCode}
  Let $q_0$, $n$, $k$, $\ell$, $\ve{t}$ and $\ve{h}$ be valid parameters, and assume that $\ell \le \frac{1}{2}(\sqrt{n} - 3)$. Let $\Csub = \TRS_{k,n}[\alphaVec, \tVec, \hVec, \etaVec] \cap \FqZ^n$. Then,
\begin{equation*}
(\Csub)^{(\star 2)} = \RS_{2k-1,n}[\alphaVec].
\end{equation*}
\end{proposition}

\begin{proof}
  \correct{
  We use the notation and the results of Proposition~\ref{prop:subcode}. This yields
\[
  \begin{aligned}
    (\Csub)^{(\star 2)} &= \Span_{\FF_{q_0}}\big(\{ \ev{X^i}{\alphaVec} \star \ev{X^j}{\alphaVec}\,:\, (i,j) \in \Ifree \}\big) \\
    &= \Span_{\FF_{q_0}} \big(\{ \ev{X^i}{\alphaVec}\,:\, i \in \Ifree + \Ifree \}\big)  \,.
  \end{aligned}
  \]
  As a consequence, the theorem holds if $\mathcal{I} + \mathcal{I} = \{0,\hdots,2k-2\}$.
  }
  
\correct{
  Notice that for valid parameters, we have $2k-1 \le n -3$ and $\mathcal{I} = \Ifree_1 \cup \Ifree_2$, where $\Ifree_1 = \{0, \dots, r-1\}$, $\Ifree_2 = \{r+\ell,\dots,k-1\}$ and $r = \lceil \frac{n+1}{\ell+2} \rceil + 2$.  We have  $\{ 0 \} + \Ifree = \{0, \dots, r-1\}$, $\Ifree_1 + \Ifree_2 = \{r+\ell,\dots,k+r-2\}$ and $\{k-1\} +  \Ifree_2 = \{k+r+\ell-1, \dots, 2k-2 \}$, hence it is clear that $\mathcal{I} + \mathcal{I}$ contains the subset
  \[
  \{0, \dots, r-1\} \cup \{r+\ell,\dots,k+r-2\} \cup \{k+r+\ell-1, \dots, 2k-2 \}\,.
  \]
  Moreover one can easily check that if $\ell \le r-1$, then $\{ r, \dots, r + \ell -1 \} \subseteq \mathcal{I}_1 + \mathcal{I}_1$. The condition $\ell \le r-1$ is always satisfied with valid parameters since $\ell < \sqrt{n} - 3$ and $r > \sqrt{n}$. Finally, the assumption $\ell \le \frac{1}{2}(\sqrt{n} - 3)$ leads us to $\ell \le \frac{k-r-1}{2}$ using constraints on valid parameters. This easily yields $\{ k+r-1, \dots, k+r+\ell-2 \} \subseteq \mathcal{I}_2 + \mathcal{I}_2$.
}
\hfill~\qed
\end{proof}


\begin{remark}
In practice, the assumption $\ell \le \frac{1}{2}(\sqrt{n} - 3)$ is not restrictive, since the decryption algorithm is effective only if $\ell \ll \log n$. 
\end{remark}

\correct{For valid parameters, we have $2k-1 \le n-3$, hence we can apply a Sidelnikov--Shestakov attack to the code $\Csub^{(\star 2)} \subseteq \FF_{q_0}^n$. This algorithm outputs}
a vector of locators $\alphaVecP \in \FqZ$ which is an affine transformation of the secret locators $\alphaVec$ (see Theorem~\ref{thm:sidelnikov-shestakov}). Formally, $\alphaVecP = a \alphaVec + b \ve{1}$ for some $a\in\FqZ\setminus\{0\}$ and $b\in\FqZ$, where $\ve{1} \mydef (1, \dots, 1) \in \FF_{q_0}^n$.

\subsubsection{Second step: from an affine to a linear transformation of the secret locators}

Lemma~\ref{lemma:equiCode} only ensures that $\TRS_{k,n}[\alphaVec, \tVec, \hVec, \etaVec] = \TRS_{k,n}[\hat{\alphaVec}, \tVec, \hVec, \hat{\etaVec}]$ if $\alphaVecH = a \alphaVec$ for a non-zero $a \in \F_{q_0}$. Therefore, given $\alphaVecP = a \alphaVec + b \ve{1}$, the search of a valid $b \in \FF_{q_0}$ such that $\alphaVecP - b \ve{1} =  a \alphaVec$ remains. Since $q_0$ is rather small, this search can be proceeded exhaustively as follows. Given $\alphaVecP$ and $b \in \F_{q_0}$, one first computes the code
\[
\mathcal{A}_b := \Span_{\FF_{q_0}}\big(\{ \evOp_{\alphaVecP - b \ve{1}}(X^i) : i \in \mathcal{I} \} \big).
\]
If $\mathcal{A}_b \subseteq \Cpub$ holds, then we have found a valid $b$ and hence a valid $\alphaVecH = \alphaVecP - b \ve{1}$. Notice that each individual test $\mathcal{A}_b \subseteq \Cpub$ can be performed in time $O(n^3)$.

\subsubsection{Third step: recovery of a valid pair $(\alphaVecH, \etaVecH)$}
\label{subsubsec:etaRecovery}

The previous steps provide a tuple $\alphaVecH \in \F_{q_0}^n$ which can be used as a vector of locators for the public TRS code. In order to determine a vector $\etaVecH \in \FF_q^n$ such that $\TRS_{k,n}[\alphaVec, \tVec, \hVec, \etaVec] = \TRS_{k,n}[\hat{\alphaVec}, \tVec, \hVec, \hat{\etaVec}]$, we use the following result.

\correct{
  \begin{lemma}
    \label{lemma:eta-bis}
    Let $1 \le \ell$, and $P(X) = \sum_{i=0}^{k-1} u_i X^i + \sum_{j=1}^\ell \eta_j u_{h_j}X^{k-1+t_j} \in \calP_{k,n}[\tVec, \hVec, \etaVec]$ such that $u_{h_j} \ne 0$. Denote by $\hat{p}_{h_j}$ and $\hat{p}_{k-1+t_j}$ the coefficients of the monomials $X^{h_j}$ and $X^{k-1+t_j}$ in $\hat{P}(X) = P(a^{-1}X)$. Then, we have 
    \[
    \hat{\eta}_j =  \eta_j a^{-(k-1+t_j-h_j)} = \frac{\hat{p}_{k-1+t_j}}{\hat{p_{h_j}}}\,.
    \]
    \end{lemma}
\begin{proof}
  This is clear from the following simple computation
  \[
  \hat{P}(X) \mydef P(a^{-1}X) =  \sum_{i=0}^{k-1} u_i a^{-i}X^i + \sum_{j=1}^\ell \eta_j u_{h_j} a^{-(k-1+t_j)} X^{k-1+t_j}\,.
  \] \hfill~\qed
\end{proof}
}

\correct{
  Hence, a vector of coefficients $\hat{\etaVec}$ such that $\TRS_{k,n}[\alphaVec, \tVec, \hVec, \etaVec] = \TRS_{k,n}[\hat{\alphaVec}, \tVec, \hVec, \hat{\etaVec}]$ can be computed as follows. Pick at random a codeword $\bfc = \ev{P}{\alphaVec} \in \Cpub = \TRS_{k,n}[\alphaVec, \tVec, \hVec, \etaVec]$. Then, interpolate $\bfc = \ev{\hat{P}}{\hat{\alphaVec}}$ as a polynomial evaluated over the vector of locators $\hat{\alphaVec}$. Notice that we have $\hat{P}(X) = P(a^{-1} X)$, thus for every non-zero coefficient $u_{h_j}$ of $P$, we obtain the coefficient $\hat{\eta}_j$ due to Lemma~\ref{lemma:eta-bis}.
}

\correct{
  It remains to be observed that, if a codeword $\bfc$ is picked uniformly at random in $\Cpub$,  the probability that  $u_j = 0$ is roughly $1/q$. Since $\ell \ll q$, a random $\bfc$ leads to the recovery of the whole vector $\hat{\etaVec}$ with high probability. Note that this procedure can be derandomized by iteratively taking each row the public matrix $\Gpub$.
}

\subsubsection{Final step: recovery of an alternative private key $(\hat{\S},\alphaVecH,\etaVecH)$}
After determining $\alphaVecH$ and $\etaVecH$, one can easily compute a matrix $\hat{\S}$ such that $\hat{\S}\GtrsH = \Gpub$. Then, $(\hat{\S},\alphaVecH,\etaVecH)$  can be  used in the proposed decryption algorithm as a valid (alternative) private key to retrieve any secret plaintext $\plain$.

\subsection{Analysis of the attack}

A summary of the attack is given in Algorithm~\ref{alg:recover}. Let us explain the notation we use there. Given a matrix $\A \in \FF_q^{k \times n}$, its transpose is represented by $\A^\top$, and  $\A^\bot$ is a matrix whose rows form a basis of the right kernel of $\A$. The reduced row echelon form of  $\A$ is denoted by $\rref(\A)$. Moreover, if $\A \in \Fq^{k\times n}$ and $\B \in \Fq^{k \times n}$ have the same rowspace, then $\D = \A \backslash \B$ denotes any solution to $\D\A=\B$. Finally, in Table~\ref{tab:notation-algo} we describe functions involved in Algorithm~\ref{alg:recover}.

\begin{algorithm}[t!]
  \caption{Key-recovery attack}
  \label{alg:recover}
  \begin{flushleft}
    \hspace*{\algorithmicindent} \textbf{Input:} $\Gpub$ \\
    \hspace*{\algorithmicindent} \textbf{Output:} $\hat{\S}, \alphaVecH$, $\etaVecH$ 
  \end{flushleft}
  
  \begin{algorithmic}[1]

    \Statex
    Step 1: recovery of some locators
    \State $\Gsub \gets \SubfieldSubcode(\Gpub) \in \FqZ^{(k-\ell)\times n}$ \label{line:Gsub}
    \State $\Gsq \gets \Square(\Gsub) \in \FqZ^{(2k-1)\times n}$\label{line:Gsq}
    \State $\alphaVecP \gets \SidelnikovShestakov(\Gsq) \in \FqZ^{n}$\label{line:alphap}

    \Statex
    \Statex
    Step 2: exhaustive search for $b$
    \ForAll{$b \in \FF_{q_0}$}\label{line:alphahStart}
    \State $\alphaVecH \gets ( \alpha^{\prime}_1-b,\hdots,\alpha^{\prime}_n-b ) \in \FqZ^{n}$
    \State $\Grrs \gets \rs(\alphaVecH) \in \FqZ^{(k-\ell)\times n}$
    \If{$\Grrs (\Gsub^{\bot})^{\top} = \0$}
    \State {\bf break}
    \EndIf
    \EndFor\label{line:alphahStop}
    
    \Statex
    \Statex
    Step 3: recovery of $\hat{\etaVec}$
    \State $J \gets \{1, \dots, \ell\}$
    \ForAll{row $\ve{r}_i$ of $\Gpub$}\label{line:etaStart}
    \State $P(X) \gets \Interpolate(\alphaVecP, \ve{r}_i) \in \Fq^{n}$
    \ForAll{$j \in J$}
    \If{$ p_{h_j} \ne 0$}
    \State  $\hat{\eta}_j \gets \frac{p_{k-1+t_j}}{ p_{h_j}} \in \Fq$
    \State $J \gets J \setminus \{ j \}$
    \EndIf    
    \EndFor
    \If{$J = \varnothing$}
    \State {\bf break}
    \EndIf 
    \EndFor\label{line:etaStop}

    \Statex
    \Statex
    Step 4: recovery of $\hat{\S}$
    \State $\Ghattrs \gets \opGTRS(\alphaVecH,\etaVecH) \in \Fq^{k \times n}$\label{line:GTRS}
    \State $\hat{\S} \gets \Ghattrs \backslash \Gpub \in \Fq^{k \times k}$\label{line:Shat}
    \State \Return $\hat{\S}$, $\alphaVecH$, $\etaVecH$
  \end{algorithmic}
\end{algorithm}

\begin{table}[t]
\correct{
\begin{tabular}{p{0.22\textwidth} p{0.67\textwidth}}
  \hline
  Function & Description \\
  \hline
  $\SubfieldSubcode$ & maps a generator matrix of $\Cpub$ to a generator matrix of the subfield subcode of $\Cpub$ \\
  $\Square$ & maps a generator matrix of $\Csub$ to a generator matrix of the code $\Csub^{(\star 2)}$\\
  $\Interpolate$ & maps vectors $(\a,\b) \in (\FF_q^n)^2$ to $P(X)$ of degree $<n$ such that $P(a_i) = b_i$ for $i=1,\hdots,n$\\
  $\rs$ & maps a vector $\a = (a_1, \dots, a_n) \in \FF_{q_0}^n$ to a matrix $\A \in \FF_{q_0}^{(k-\ell) \times n}$ whose rows are $(a_1^j, \dots, a_n^j)$ for each $j \in \mathcal{I} =  \{0, \dots, k-1\} \setminus \{h_1, \dots, h_\ell \}$.\\
  $\SidelnikovShestakov$ & implements a Sidelnikov--Shestakov attack, which takes a generator matrix $\G$ of a RS code as input, and returns a vector of locators $\alphaVec'$ describing the code \\
  $\opGTRS$ & maps the vectors $\alphaVecH$ and $\etaVecH$ to the generator matrix $\GtrsH$ of the corresponding TRS code \\
  \hline
\end{tabular}
\caption{\label{tab:notation-algo}List of functions used in Algorithm~\ref{alg:recover}.}
}
\end{table}

\correct{
\begin{theorem}
  \label{thm:main-thm}
  Given any generator matrix $\Gpub$ of a TRS code $\Cpub = \TRS_{k,n}[\alphaVec, \tVec, \hVec, \etaVec] \subseteq \FF_q^n$, Algorithm~\ref{alg:recover} retrieves a tuple $(\hat{\S}, \hat{\alphaVec}, \hat{\etaVec})$ such that the matrix  $\hat{\S} \GtrsH$ generates $\Cpub$ in $O(\max\{q_0, 2^\ell, n\} \cdot n^3)$ operations over $\FF_q$.
\end{theorem}
}
\begin{proof}
  The correctness of Algorithm~\ref{alg:recover} was proved in Section~\ref{sec:key-recovery}. Let us now provide details about the complexity of Algorithm~\ref{alg:recover}.
\begin{itemize}
\item Line~\ref{line:Gsub}: The computation of $\Gsub \in \FqZ^{(k-\ell)\times n}$ requires $O(n^2(k+n))\subseteq O(n^3)$ operations in $\Fq$ and $O(n^2(2^{\ell}(n-k)+n))\subseteq O(2^\ell n^3)$ operations in $\FqZ$.
\item Line~\ref{line:Gsq}: The computation of $\Gsq \in \FqZ^{(2k-1)\times n}$ can be performed in time $O(n^4)$. Informally, one needs to find a basis of the space generated by the set $\{ \ve{g}_{i,j} := ({\Gsub})_i \star ({\Gsub})_j, 1 \le i, j \le \dim \Csub \}$. This basis can be built iteratively; updating the basis with a new element  costs $O(n^3)$ operations in $\F_{q_0}$ and must be done $O(n)$ times, and rejecting candidates costs $O(n^2)$ operations in $\F_{q_0}$ and must be done $O(n^2)$ times.
\item Line~\ref{line:alphap}: Applying the $\SidelnikovShestakov$ function on $\Gsq \in\FqZ^{(2k-1)\times n}$ requires $O( (2k-2)^4 + (2k-2)n) \subseteq O(n^4)$ operations in $\FqZ$~\cite{sidelnikov1992attack}.
\item Line~\ref{line:alphahStart} to Line~\ref{line:alphahStop}:  The computation of $\alphaVecH \in \FqZ^{n}$ requires $O(n)$ operations in $\FqZ$; building $\Grrs \in \FqZ^{(k-\ell)\times n}$ needs $O((k-\ell)n)$ operations in $\FqZ$; matrix multiplication $\Grrs (\Gsub^{\bot})^{\top}$ needs $O((k-\ell)(n-k+\ell)n) \subseteq O(n^3)$ operations in $\FqZ$ ($\Gsub^{\bot}$ was already computed in Line~\ref{line:Gsub}). In the worst case, the previous sequences of computations have to be performed $q_0$ times. Hence these steps require $O(q_{0}n^3)$ operations in $\FqZ$.
\item Line~\ref{line:etaStart} to Line~\ref{line:etaStop}: In the worst case, $\ell \cdot k$ interpolations have to be performed, requiring $O(\ell k n^2) \subseteq O(\ell n^3)$ operations in $\Fq$.
 \item Line~\ref{line:GTRS}: Computation of $\Ghattrs\in \Fq^{k\times n}$ needs $O(kn)\subseteq O(n^2)$ operations in $\Fq$.
 \item Line~\ref{line:Shat}: Computation of $\hat{\S} \in \Fq^{k \times k}$ by a reduction to row echelon form  of the matrix $\begin{mymatrix} \Ghattrs^{\top}& \Gpub^{\top}  \end{mymatrix} \in \Fq^{n\times 2k}$ needs $O(n^2(2k)) \subseteq O(n^3)$ operations in $\Fq$. \hfill~\qed
\end{itemize}
\end{proof}

In practice, $\ell$ and $q_0 = q^{1/2^{\ell}}$ must be chosen to be small (for instance, $\ell =1$ and $q_{0} = n+1 = 2^8$ were proposed in~\cite{beelen2018structural}) in order to obtain an efficient decryption algorithm and keys of moderate size. Hence, for practical parameters Algorithm~\ref{alg:recover} has a complexity in $O(n^4)$ and thus recovers a valid private key in polynomial time.

Our attack is implemented in the computer algebra system SageMath v8.7~\cite{sagemath} and is available at \url{https://bitbucket.org/julianrenner/trs_attack}. Although the implementation is not optimized, it recovers a valid private key within a few minutes for the proposed parameters, see Table~\ref{tab:para}.

\setlength{\tabcolsep}{11pt}
\begin{table*}[t!]
  \renewcommand{\arraystretch}{1.1} 
	\begin{center}
		\begin{tabular}{l|c|c|c|c|c||c}
                   $q_0$ & $n$ & $k$ & $\ell$ & $\wt(\e)$ & \makecell{Claimed \\security level} & \makecell{Runtime of\\ Algorithm~\ref{alg:recover}}   \\
                  \hline \hline
                  $2^8$ & $255$ & $117$ & $1$ & $83$ & $128$ bits\phantom{$^{*}$}  & 133 seconds \\
                  $2^8$ & $255$ & $117$ & $2$ & $83$ & $128$ bits\phantom{$^{*}$}  & $141$ seconds \\
                  $2^9$ & $511$ & $200$ & $3$ & $192$ & $196$ bits\phantom{$^{*}$} &  $2260$ seconds \\
                  $2^9$ & $511$ & $170$ & $3$ & $217$ & $256$ bits\phantom{$^{*}$}  & $1532$ seconds 
		\end{tabular}
              \end{center}
        \caption{Experimental results obtained by averaging several runtimes of Algorithm~\ref{alg:recover} on an Intel(R) Core(TM) i7-7600U CPU @ 2.80GHz. The first row refers to parameters proposed by the designers of the system. Remaining security levels were computed according to formulae given in~\cite{beelen2018structural}.}
	\label{tab:para}
\end{table*}

\section{Discussion and open questions}
\label{sec:discussion}

\subsection{Repairing the cryptosystem?}

After a notification of this attack, the authors of~\cite{beelen2018structural} described a possible fix of the system, in which a modified version of the generator matrix is made public. The idea is to multiply the generator matrix $\Gpub$ from the right by a diagonal matrix with non-zero entries $\ve{y} = (y_1, \dots, y_n) \in (\F_q \setminus \{ 0 \})^n$, such that the $\F_{q_0}$-subfield subcode of the vector space spanned by the rows of $\Gpub$ is not contained in a RS code. This clearly prevents a direct application of our attack.

Nevertheless, we would like to point out that this possible repair might not fix the inherent weaknesses of the cryptosystem. In fact, the subfield subcode of a GRS code $\ve{y} \star \RS_{k,n}[\alphaVec]$ is a so-called \emph{alternant code} ${\rm Alt}_{k',n}[\alphaVec, \ve{y}] \subseteq \F_{q_0}^n$, which also admits an algebraic description. As a consequence, it seems very plausible that the security of the proposed repaired cryptosystem can be reduced to the security of a McEliece-like cryptosystem using the subfield subcode ${\rm Alt}_{k',n}[\alphaVec, \ve{y}]$.

One can then notice that the parameters proposed by the authors of~\cite{beelen2018structural} are far below those considered as secure for alternant codes. For instance, BIG QUAKE~\cite{BigQuake} and Classic McEliece~\cite{ClassicMcEliece} (both are unbroken candidates for the NIST standardization call on post-quantum cryptography) use alternant codes with a length and dimension of several thousands, while in the proposed parameters for the TRS codes, we have $n = 255$ and $k = 117$ with a field size $q_0 = 2^8$. Algebraic attacks as developed in~\cite{FaugereOPT10, FaugereOPPT16} should then be considered as a potential threat. One should also mention the recent attack on the  cryptosystem DAGS~\cite{DAGS} based on alternant codes, performed by Barelli and Couvreur~\cite{BarelliC18}. Informally, the authors of~\cite{BarelliC18} manage to derive an alternant code with much smaller parameters from the public code, allowing the last step of the key recovery algorithm --- which is exponential in the involved parameters --- to remain feasible due to the small size of the derived alternant code.

Finally and most crucially, one can question the possible benefit to consider codes whose security might not be better than those based on alternant codes (for which cryptosystems have been designed and studied), but which suffer from larger key sizes and much less efficient decoding algorithms.

\subsection{On the rank metric version of the cryptosystem}

In~\cite{puchinger2018twisted} a modified version of the previous system was proposed, based on a subfamily of twisted Gabidulin codes. The idea is to consider a variant of the GPT cryptosystem~\cite{gabidulin1991ideals}, where twisted Gabidulin codes are used instead of (subcodes of) Gabidulin codes. Although we do not claim to have a proper attack on the system, let us show some potential weaknesses that could be analyzed in a future work.

The GPT cryptosystem can be viewed as an analogue of the McEliece cryptosystem, using rank metric codes instead of Hamming metric codes. We refer to~\cite{OverbeckPhD2007} for more details about rank metric codes and variants of the GPT cryptosystem. Let us give a short overview of the latter.

Let $\Fp \subset \Fq$ and  $\Gamma \subseteq \{ \mathcal{C} \subseteq \Fq^{n-t}, \dim \mathcal{C} = k \}$ be a family of rank metric codes. The GPT cryptosystem works as follows:
  \begin{itemize}
  \item \emph{Key generation:} Alice generates a secret generator matrix $\ve{G} \in \Fq^{n-t}$ for a code $\mathcal{C}$ randomly chosen in $\Gamma$. Then she computes a public key $\Gpub = \S [\X | \ve{G}] \P$, where the matrices $\S \in \GL_k(\Fq)$, $\X \in \Fq^{k \times t}$ of rank $s \le t$,  and $\P \in \GL_n(\Fp)$ are chosen uniformly at random and kept secret.
  \item \emph{Encryption:} given a plaintext $\plain \in \Fq^k$, Bob computes the ciphertext $\cipher = \plain \Gpub + \e$, where $\e \in \Fq^n$ is a random error with small rank over $\Fp$ (the rank of the error is such that $\e$ can be decoded in $\mathcal{C}$).
    \item \emph{Decryption:} Alice decodes the last $n-t$ coordinates of  $\cipher \P^{-1}$ in the code $\mathcal{C}$ and retrieves $\plain$.
  \end{itemize}
  In most variants of the GPT cryptosystem, $\Gamma$ is a (sub-)family of Gabidulin codes~\cite{Gabidulin_TheoryOfCodes_1985}
  \[
  \Gab{\alphaVec}{k, n-t} = \Span_{\FF_q} \big\{ \ev{X^{[i]}}{\alphaVec},\; i=0,\dots,k-1 \big\},
  \]
  where $\alphaVec \in \FF_q^{n-t}$ are $\FF_p$-linearly independent, and $X^{[i]} := X^{p^i}$. Polynomials with monomials only of the form $X^{[i]}$ are called $p$-polynomials, or linearized polynomials. In~\cite{puchinger2018twisted}, the authors proposed to define $\Gamma$ as the subfamily of twisted Gabidulin codes
  \[
  \GmultS = \Big\{ \ev{f}{\alphaVec} \, : \, f \in \Big\{ \sum_{i=0}^{k-1} f_i X^{[i]} + \sum_{j=1}^{\numTwists} \eta_j f_{h_j} X^{[k-1+t_j]} : f_i \in \Fq \Big\} \Big\},
  \]
  where  $\eta_i$ are chosen in the chain of subfields $\Fqsi{0} \subset \Fqsi{1} \subset \hdots \subset \Fqsi{\ell} = \Fq$, and $(\alpha_1, \dots, \alpha_{n-t}) \in \Fqsi{0}^{n-t}$ are $\Fp$-linearly independent, similar to the case of TRS codes.

Our claim is that the code $\mathcal{C}_{\rm pub}$ generated by $\Gpub$ also admits structured subfield subcodes which could be used to attack the system. Indeed, one can prove that the last $n-t$ coordinates of $(\mathcal{C}_{\rm pub} \cap \Fqsi{0}^n)\P^{-1}$ form a subcode of the Gabidulin code $\Gab{\alphaVec}{k,n-t} \subseteq \Fqsi{0}^{n-t}$ of rather small codimension. Applying variants of Overbeck's attacks --- e.g. in~\cite{Overbeck-StructuralAttackGPT} --- might lead to the recovery of a linear transformation of $\alphaVec$ and thus a structural attack on the public key close to the one presented in this paper.

For a code $\mathcal{A} \subseteq \FF_q^n$ and $f \ge 0$, let $\mathcal{A}^{[1]} \mydef \{ (a_1^{[f]}, \dots, a_n^{[f]}), \a \in \mathcal{A} \}$, and 
\begin{equation*}
  \Lambda_{f} (\mathcal{A}) \mydef  \mathcal{A} + \mathcal{A}^{[1]} + \dots + \mathcal{A}^{[f]}\,.
\end{equation*}
In fact, we observe in simulations that for  $f = n-k-t-1$, if $\Lambda_{f}(\mathcal{C}_{\rm pub} \cap \Fqsi{0}^n)$ has dimension $n-1$, then one recovers an $\F_p$-linear transformation $\alphaVecH$ of $\alphaVec$, as well as a full-rank matrix $\Phat \in \F_p^{n \times n}$, by applying~\cite[Algorithm 3.5.1]{OverbeckPhD2007} to a generator matrix of $\mathcal{C}_{\rm pub} \cap \Fqsi{0}^n$. Then, the coefficients $\etaVecH$ are determined by interpolation of the last $n-t$ positions of the rows of $\Gpub \Phat^{-1}$ with $p$-polynomials of $p$-degree smaller than $n$, similar to Section~\ref{subsubsec:etaRecovery}. Finally, one chooses $\Shat$ such that
\begin{equation*}
\Shat \Ghat = \big(\Gpub \Phat^{-1}\big)_{[t+1:n]},
\end{equation*}
where subscript $[t+1:n]$ refers to the last $n-t$ positions of $\Gpub \Phat^{-1}$ and $\Ghat$ is a generator matrix of $\GmultHatS$. Clearly, ($\Shat,\alphaVecH, \etaVecH, \Phat) $ is then a valid private key.

Further simulations show that if $\X$ has full $\Fq$-rank and $t$ is small, then the code $\Lambda_{f}(\mathcal{C}_{\rm pub} \cap \Fqsi{0}^n)$ has a dimension $n-1$ with high probability. However, if  $t$ is large or if $\X$ has $\Fq$-rank smaller than $t$, then $\Lambda_{f}(\mathcal{C}_{\rm pub} \cap \Fqsi{0}^n)$ has dimension smaller than $n-2$ and this straightforward approach fails.

Since a precise analysis of the potential weakness of system proposed in~\cite{puchinger2018twisted} is out of the scope of this paper, we leave it as an open problem for future research.

\section{Conclusion}\label{sec:conclusion}

This paper presents an efficient key-recovery attack on the McEliece cryptosystem based on a subfamily of TRS codes. The attack does not contradict the structural properties presented in~\cite{beelen2018structural}, but recovers the structure of a \emph{subfield subcode} of the public TRS code, which enables us to determine a description of the supercode. This attack retrieves a valid private key from the public key for all practical parameters in $O(n^4)$ field operations. This is formally proved, and confirmed by experimental results: one retrieves a valid private key for a claimed security level of $128$ bits within a few minutes. 
In addition, the security of the system after an attempt to repair it is discussed, as well as potential ways to adapt our attack to the rank metric variant of the considered system.

 \correct{The subfield subcode approach presented in this paper is unique, in the sense that a widespread idea considers the restriction of codes to subfields as a way to break their structure. However, our cryptanalysis proves that subfield subcodes --- as well as punctured codes and shortened codes --- must also be taken into account when trying to assert the security of McEliece-like cryptosystems.}


\section*{Acknowledgements}
This work was done while the second author was visiting the Institut de Recherche Mathématique de Rennes (IRMAR), Université de Rennes 1, France.

The first author is funded by the French \emph{Direction Générale l'Armement}, through the \emph{Pôle d'excellence cyber}.

This project has received funding from the European Research Council (ERC) under the European Union’s Horizon 2020 research and innovation programme (grant agreement No~801434).

We would like to thank Antonia Wachter-Zeh (TUM) for fruitful discussions and Oliver De Candido (TUM) for his comments that helped to improve the manuscript. We would further like to thank the authors of the proposed cryptosystem~\cite{beelen2018structural} for validating our attack and pointing out a possible repair of the system with respect to our attack.

\bibliographystyle{alpha}
\bibliography{main}

\end{document}

%% file: defs.tex
\usepackage[utf8]{inputenc}
\usepackage[english]{babel}
\usepackage[T1]{fontenc}
\usepackage{amsmath, amssymb, amsbsy, amsthm}
\usepackage{mathdots}
\usepackage{mathtools}
\usepackage{xspace}
\usepackage[noend]{algpseudocode}
\usepackage{algorithm}
\usepackage{algorithmicx}
\algdef{SE}[DOWHILE]{Do}{doWhile}{\algorithmicdo}[1]{\algorithmicwhile\ #1}%
\usepackage{color}
\usepackage{cite}
\usepackage{booktabs}
\usepackage{verbatim}
\usepackage{url}
\usepackage{lipsum}
\usepackage{enumitem}
\usepackage{colortbl}
\usepackage{makecell}
\usepackage{hyperref}

\hypersetup{
  colorlinks   = true, 
  urlcolor     = blue, 
  linkcolor    = blue, 
  citecolor    = red 
}

\usepackage[margin=3cm]{geometry}
\newtheorem{theorem}{Theorem}
\newtheorem{lemma}[theorem]{Lemma}

\newtheorem{proposition}[theorem]{Proposition}

\newtheorem{remark}[theorem]{Remark}

\newtheorem{definition}{Definition}

\newenvironment{mymatrix}{\begin{pmatrix}} {\end{pmatrix} }

\usepackage{xcolor}

\newcommand{\ve}[1]{\ensuremath{\boldsymbol{#1}}}

\newcommand{\NN}{\mathbb{N}}      
\newcommand{\F}{\ensuremath{\mathbb{F}}}

\newcommand{\Gsq}{\G_{\text{sq}}}
\newcommand{\Grrs}{\G^{\prime}}
\newcommand{\Ghattrs}{\hat{\G}_{\text{TRS}}}

\newcommand{\FqZ}{\mathbb{F}_{q_0}}
\newcommand{\FqO}{\mathbb{F}_{q_1}}
\newcommand{\Fqi}{\mathbb{F}_{q_i}}
\newcommand{\Fqii}{\mathbb{F}_{q_{i-1}}}
\newcommand{\Fqell}{\mathbb{F}_{q_{\ell}}}
\newcommand{\Fq}{\ensuremath{\mathbb{F}_q}}

\newcommand{\SubfieldSubcode}{{\tt SubfieldSubcode}}
\newcommand{\Square}{{\tt Square}}
\newcommand{\Interpolate}{{\tt Interpolate}}
\newcommand{\SidelnikovShestakov}{{\tt SidelShest}}
\newcommand{\rs}{{\tt GenSub}}
\newcommand{\rref}{{\tt rref}}
\newcommand{\opGTRS}{{\tt GTRS}}

\DeclareMathOperator{\Span}{Span}
\DeclareMathOperator{\GL}{GL}

\newcommand{\wt}{w_{\text{H}}}
\newcommand{\cipher}{\ve{y}}
\newcommand{\plain}{\ve{m}}
\newcommand{\plainT}{\tilde{\ve{m}}}

\newcommand{\FF}{\mathbb{F}}
\newcommand\mydef{\coloneqq}

\newcommand{\bfA}{\ve{A}}
\newcommand{\bfG}{\ve{G}}

\newcommand{\bfc}{\ve{c}}
\newcommand{\bfe}{\ve{e}}

\newcommand{\calC}{\mathcal{C}}
\newcommand{\calP}{\mathcal{P}}

\newcommand{\RS}{\mathrm{RS}}
\newcommand{\TRS}{\mathrm{TRS}}

\newcommand{\correct}[1]{#1}

\newcommand{\tVec}{\ve{t}}
\newcommand{\hVec}{\ve{h}}
\newcommand{\etaVec}{\ve{\eta}}
\newcommand{\alphaVec}{\ve{\alpha}}

\newcommand{\Gtrs}{\G_{\alphaVec,\tVec,\hVec,\etaVec}}

\newcommand{\GtrsH}{\G_{\alphaVecH,\tVec,\hVec,\etaVecH}}

\newcommand{\numTwists}{\ell}

\newcommand{\etaVecH}{\hat{\ve{\eta}}}
\newcommand{\alphaVecH}{\hat{\ve{\alpha}}}

\newcommand{\alphaVecP}{\ve{\alpha}^{\prime}}

\newcommand{\Ifree}{\mathcal{I}}

\DeclareMathOperator{\evOp}{ev}
\newcommand{\ev}[2]{\evOp_{#2}(#1)}

\newcommand{\0}{\ve{0}} 
\renewcommand{\S}{\ve{S}}

\renewcommand{\a}{\ve{a}}
\renewcommand{\b}{\ve{b}}
\newcommand{\e}{\ve{e}}

\renewcommand{\c}{\ve{c}}
\renewcommand{\e}{\ve{e}}

\newcommand{\G}{\ve{G}}
\renewcommand{\P}{\ve{P}}

\newcommand{\Gpub}{\ve{G}_\mathrm{pub}}

\newcommand{\Cpub}{\mathcal{C}_\mathrm{pub}}
\newcommand{\Csub}{\mathcal{C}_\mathrm{sub}}
\newcommand{\Gsub}{\ve{G}_\mathrm{sub}}

\newcommand{\A}{\ve{A}}
\newcommand{\B}{\ve{B}}
\newcommand{\D}{\ve{D}}

\newcommand{\Fp}{\mathbb{F}_{p}}

\newcommand{\GabCode}{\mathcal{G}}
\newcommand{\Gab}[2]{\GabCode_{#2}[#1]}
\newcommand{\TGab}[3]{\GabCode_{#2}^{#3}[#1]}

\newcommand{\GmultS}{\TGab{\alphaVec,\tVec,\hVec,\etaVec}{n-t,k}{}}
\newcommand{\GmultHatS}{\TGab{\alphaVecH,\tVec,\hVec,\etaVecH}{n-t,k}{}}

\newcommand{\Fqsi}[1]{\mathbb{F}_{q_{#1}}}

\newcommand{\X}{\ve{X}}

\newcommand{\Phat}{\hat{\P}}
\newcommand{\Shat}{\hat{\S}}
\newcommand{\Ghat}{\hat{\G}}

